\documentclass[runningheads]{llncs}
\usepackage{graphicx}
\usepackage{cite}
\usepackage{amsmath,amssymb,amsfonts}
\usepackage{algorithmic}
\usepackage{textcomp}
\usepackage{subfig}

\newtheorem{Theorem}{Theorem}
\newtheorem{Proposition}{Proposition}
\newtheorem{Comment}{Comment}

\begin{document}
\title{Optimal Energy Distribution with Energy Packet Networks}
%
%
\author{Yunxiao Zhang }
\authorrunning{Y. Zhang}
%
\institute{Imperial College London \\
\email{yunxiao.zhang15@imperial.ac.uk}\\
}
\maketitle              
\begin{abstract}
We use Energy Packet Network paradigms to investigate energy distribution problems in a computer system with energy harvesting and storages units. Our goal is to minimize both the overall average response time of jobs at workstations and the total rate of energy lost in the network. Energy is lost when it arrives at idle workstations which are empty. Energy is also lost in storage leakages. We assume that the total rate of energy harvesting and the rate of jobs arriving at workstations are known. We also consider a special case in which the total rate of energy harvesting is sufficiently large so that workstations are less busy. In this case, energy is more likely to be sent to an idle workstation. Optimal solutions are obtained which minimize both the overall response time and energy loss under the constraint of a fixed energy harvesting rate.

\keywords{Energy distributions, Energy harvesting, Energy Packet Networks, G-networks, Optimization, Product-form solution, Renewable energy.}
\end{abstract}
\section{Introduction}

Although conventional fossil fuels are still the primary energy source in the world today, demands on alternative energy sources which are renewable, environmentally friendly trigger many investigations into energy harvesting and energy storage systems.   
There are recent researches on understanding challenges and opportunities in energy harvesting and storage. 
The paper \cite{Challenges, ESS_grid} reviews the technologies of energy storage systems and the benefits regarding finance and reliability, which energy storages can bring in power systems. It also emphasizes the importance to develop small to medium-sized energy storages near to renewable energy sources which are in a distributed manner in the near future. Energy harvesting and storage are not only playing a critical role in power systems, but they also have many applications in computer and communication systems that support the Internet of Things (IoT). 

The critical increase in energy consumption in information and communication technologies (ICT) \cite{Caseau} creates a huge need for the parsimonious energy usage in ICT systems. Thus early research \cite{Power1,Power2} has addressed the need for power-aware routing in wireless networks where the number of hops that a packet traverses, including in simple opportunistic networks \cite{Opportunistic}, also increase the network's energy consumption \cite{Packet}. The interplay of routing and energy consumption in wired networks has been studied in \cite{Morfo,Toktam}. Energy savings in Cloud Computing has been studied in \cite{Berl}, where adaptive task scheduling methods can also improve energy savings \cite{Lent-Task,Wang18}. 

A complex interconnected distributed computer system is considered in \cite{gelenbe2015central}. To improve sustainability, the system is operated with energy harvesting nodes and energy storage units. It investigates distributed and centralized energy storage units and compares the overall response time of data packets in two types of storages. In \cite{EH_optimal, OptimalPacket_Yang}, it considers energy harvesting in wireless communication systems comprised of rechargeable nodes. It studies how to control and optimize the energy consumption and quality of service (QoS) subject to finite energy storage capacity and causality constraints. In \cite{paperOne}, it introduces a framework of energy cooperation and energy sharing in communication networks where users transmit messages using energy harvested from nature. 
Authors consider several multi-user scenarios, where energy can be transferred from one user to another through a separate wireless transfer unit. Optimum transmit power and energy transfer policies are determined by a Lagrangian approach and the two-dimensional directional water-filling algorithm. In \cite{paperTwo}, authors consider an energy harvesting, two user cooperative Gaussian multiple access channel (MAC), where both users harvest energy from nature. Authors study two scenarios with this model. In the first scenario, both data and energy arrive intermittently at the nodes. In the second scenario, the data packets are available at the beginning of the transmission. Users can cooperate at the physical layer (data cooperation) and the battery level (energy cooperation). In \cite{paperFour}, authors investigate a queueing model of wireless communication base stations (BSs) which is fully powered by renewable energy sources. The BS can dynamically adjust its coverage area, the number of mobiles with which it communicates, offered rate and energy consumption. A summary of recent contributions in energy harvesting wireless communications and wireless energy  transfer from the perspectives of communication theory, signal processing, information theory and wireless networking is in \cite{paperThree}.

In this paper, we focus on a specific smart energy distribution problem in computer systems with energy harvesting and storages. We consider a system consisted of workstations (WSs), each of which is powered by one interconnected energy storage (ES). ESs are charged by one energy harvesting unit with an intermittent energy source, such as wind turbines or photovoltaic solar systems. In the energy distribution problem, we assume energy cannot be transferred between ESs and WSs do not exchange jobs among them. We also assume that the system receives energy at a fixed rate. The problem is to choose an optimal energy distribution which determines portions of the fixed harvested energy rate that are sent to WSs, so as to minimize the overall average response time of jobs and energy loss. Energy may be lost in ES leakages. We also assume that energy is lost if it is sent to an idle WS which has no job.

We use Energy Packets Networks (EPNs) \cite{gelenbe2011energy,gelenbe2012energy,Compsac2014,Fourneau} as an approach to investigate these considerations. The EPN is a discretized state-space framework of a family of queueing model known as G-networks \cite{gelenbe1993trigger,Multiple}. A simplified model of G-networks was first proposed in the early 1990s \cite{Stable}. It considers an open network of $v$ queues that have mutually independent and identically distributed (i.i.d.) exponential service time of rates $r(1),\dots,r(v)$. 

In the early G-network papers \cite{Gelenbe1994unifying}, there are two types of customers: positive and negative customers that arrive at $i^{th}$ queues according to a Poisson process of rates $\Lambda_i^+$ and $\Lambda_i^-$ respectively.  Queues can only be constituted by positive customers, while one negative customer arriving at a non-empty queue can either move a positive customer to another queue or remove a positive customer out of the network. Then the negative customer disappears. Moreover, a product-form solution exists in G-networks. It is a joint probability distribution of the number of positive customers at queues in the steady state. Up to today, a series of G-network papers regarding G-networks with batch removals \cite{gelenbe_1993BatchRemove}, multi-class 
G-networks \cite{Multiple}, triggers \cite{gelenbe1993trigger}, resets \cite{reset} and adders \cite{Adder} are published.

Since EPN models are special cases of G-networks, we can represent EPN models as G-networks. In EPN models, jobs arrive at servers or WSs, each of which is powered by ESs. Energy is represented in the form of energy packets (EPs), each of which is a fixed amount of energy in Joules. EPs at ESs and jobs at WSs are positive customers that constitute ESs and WSs respectively. However, EPs can be moved from ESs to WSs. At WSs, EPs become negative customers that can move one job to another WS or remove jobs out of the network. In this energy distribution problem, the EPN is modeled as a special case of G-networks with batch removal where one EP can be used to execute a batch of jobs and remove executed jobs out of the system. The size of a batch follows a probability distribution which is known in this paper.  Moreover, we assume that one ES is interconnected with only one WS.  

Recently, EPN models and G-networks are used in various applications. In \cite{gelenbe2010frameworkpacket}, it proposes an on-line routing algorithm to reduce energy consumption and overall response time of data packets in packet networks. Moreover, a set of utility functions regarding QoS and energy storages are given in \cite{Access}. By means of the product-form solution, the probability that the number of EPs at energy storages is greater than a number $k$, is given in the utility function. Moreover, EPN is also used in \cite{NOLTA} to investigate the backhaul of mobile networks that operate with intermittent renewable energy. EPN models also generate interest in modeling and optimizing sensor networks powered by energy harvesting units \cite{Fourneau, Marin,yin2017EPN, Yasin, Kadi}.
Furthermore, another approach known as power packet is proposed in \cite{takahashi2012estimation,takahashi2013design}. Power and information data can be carried in power packets. It is tested by hardware systems which provide insights into how to build a smart electricity distribution system.

Our contributions are as follows:
\begin{enumerate}
\item The developed EPN model is a special case of a family of queening network known as G-Networks that are developed in a series of papers starting around 1990 up to today. In previous researches, it is  assumed that processing at the workstation can only occur if one energy packet arrives. The workstation either sends one job which has just completed to another workstation for more processing or removes one job from the network. Hence, one single job can consume only one EP at a time. In this paper, we consider a more complex  EPN model in which one single energy packet can be consumed to process a variable number of jobs (instead of a single job).
\item A new composite cost function is proposed. In the cost function, we consider both the overall average response time of jobs and the total rate of energy lost in the network.
\item Solutions of the composite cost function can easily be obtained by solving a set of equations analytically.   
\end{enumerate}

The paper is organized as follows. In Section \ref{EPN_Model}, we introduce the EPN model and its product-form solution which is the joint probability distribution of the numbers of jobs and EPs at WSs and ESs respectively. A related optimization problem of energy distribution is discussed Section \ref{opt_problem_with_GA}. We concern both the overall response time of jobs and the energy loss. In Section \ref{ana_problem_with_GA}, we analyze the optimization problem under a geometric assumption. Under the assumption, analytical solutions are obtained in the case of $N=2$ (two pairs of WSs and ESs) and the case of $N \geq 3$ (more than two pairs of WSs and ESs). A proposition is also given to guarantee the existence of a globally optimal solution. A special case of the optimization problem is discussed in Section \ref{problem_special} where the total rate of harvesting energy is assumed to be sufficiently large. An analytical solution is also obtained for the special case. Numerical examples are presented in each case to illustrate obtained optimal solutions. Conclusions are given in Section \ref{conclusions}.
 
\section{Energy Packet Network Model}
\label{EPN_Model}
We consider an EPN model which is an open network with  $v=2N$ of queues. Queues $1,\dots,N$ are WSs and queues $N+1,\dots,2N$ are ESs. Note that $N$ is the number of pairs of one WS and one ES.
Hence, the $i^{th}$ queue is the $i^{th}$ WS and
the $(i+N)^{th}$ queue is the $i^{th}$ ES where $i \in \{1,\dots,N \}$.
Each WS is interconnected with an ES which provides energy to this WS. 
However, we assume that EPs and jobs cannot be moved between ESs and WSs respectively.
Moreover, there are two Classes of customers in the network. Class 1 customers are jobs or tasks which are needed to be executed, while Class 2 customers are EPs which correspond to the amount of energy required to execute one or more than one jobs.   

There are also two Types of customers: \textit{positive} or \textit{negative customers}. Other Types of customers, e.g., \textit{triggers} \cite{gelenbe1993trigger}, \textit{resets} \cite{reset} and \textit{adders} \cite{Adder} are not used in this EPN model. 
\begin{enumerate}
\item Positive customers are normal customers which accumulate at queues. They are needed to be executed. After finishing execution, positive customers may be moved to other queues or removed from the network. Note that we assume positive customers cannot be moved in this paper. 
\item Negative customers do not request any service. When negative customers of Class $c$ arrive at a queue, they will instantaneously remove a batch of positive customers of Class $c'$. Then negative customers disappear. For example, an EP is instantaneously consumed to execute a batch of jobs when the EP arrives at a WS. Jobs are removed from the WS after finishing execution. Moreover, negatives customers have no effects on an empty queue. Negative customers merely disappear if they arrive at an empty queue. 
\end{enumerate}

For the $i^{th}$ WS denoted $W_i$, it has following manner:
\begin{enumerate}
\item Positive customers at $W_i$ can only be jobs. External jobs arrive at $W_i$ according to a Possion process at the rate of $\lambda_i$ jobs/sec. Moreover, the state of $W_i$ is represented by the number of jobs at $W_i$. At time $t$, it is denoted as $K_i(t)$.
\item The interconnected ES will send EPs to $W_i$. These EPs are negatives customers when they arrive at $W_i$. At time $t$, one EP is instantaneously consumed to execute a batch of jobs (positive customers at $W_i$) with size $\max[K_i(t),b_i]$. The $b_i$ is a random variable with probability distribution,
\begin{equation}
\pi_{i,s}= \Pr [b_i=s],~ s=1,2,\dots 
\end{equation} 
and 
\begin{equation}
\sum_{s=1}^\infty \pi_{i,s}=1 \mbox{ for all } i=1,\dots,N.
\end{equation}
After finishing execution, jobs are removed from $W_i$ and the EP disappears. Hence, at next time instance $t^+$, the state of $W_i$ is 
\begin{equation}
K_i(t^+)=
\begin{cases}
0, & \mbox{ if } b_i \geq K_i(t) \\
K_i(t)-b_i, & \mbox{ otherwise}
\end{cases}.
\end{equation}
\item If $W_i$ is empty at time $t$, such that $K_i(t)=0$, EPs have no effect on $W_i$. They are merely consumed to maintain $W_i$'s operation. Then EPs disappear.  
\end{enumerate}

For the $i^{th}$ ES denoted $E_i$, it has following manner:
\begin{enumerate}
\item The $i^{th}$ ES is interconnected to the $i^{th}$ WS, such that $W_i$ is powered by $E_i$. 
\item Positive customers at $E_i$ can only be EPs. External EPs arrive $E_i$ according to a Poisson process at the rate of $\gamma_i=\gamma p_i$ EPs/sec. The $\gamma$ denotes the total rate of EPs which are harvested by the network. The $p_i \in [0,1]$ are  probability variable, such that $\sum_{i=1}^N p_i =1$. Moreover, the state of $E_i$ is also represented by the number of EPs at $E_i$. At time $t$, it denotes $B_i(t)$.
\item When $B_i(t)>0$, EPs may be lost with an exponential service time at the rate of $\delta_i$ EPs/sec, or EPs may be sent to the interconnected WS with an exponential service time at the rate of $w_i$ EPs/sec. At time $t^+$, the state of $E_i$ becomes
\begin{equation}
B_i(t^+)=B_i(t)-1.
\end{equation}
\end{enumerate}

The considered EPN system is schematically presented in  Figure \ref{fig1}.
\begin{figure}[h]
	\centering
	\includegraphics[width=0.5\linewidth]{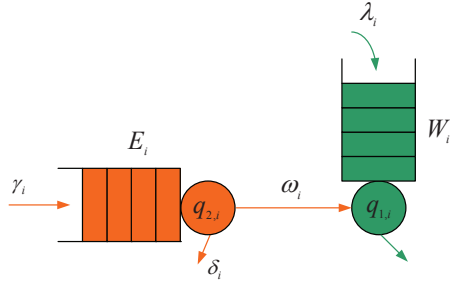}  
	\caption{Schematic representation of the pair of the $i^{th}$ WS and the $i^{th}$ ES. }
	\label{fig1}
\end{figure}

\subsection{The product-form solution of the EPN model}
The proposed EPN model can be analyzed as a G-network with batch removal. We recall a \textit{product-form solution} (PFS) which can help us to analyze the EPN model. 

Regarding to the G-network model in \cite{gelenbe_1993BatchRemove}, we consider following \textit{traffic equations}:
\begin{align}
&\begin{cases}
\Lambda_{1,i}^+&=\lambda_i,\\
\Lambda_{1,i}^-&=q_{2,i+N} w_i \big[ \frac{1-\sum_{s=1}^\infty q_{1,i}^s \pi_{i,s}}{1-q_{1,i}} \big], \\
r_{1,i}&=0\\
\Lambda_{2,i+N}^+ &=\gamma_i,  \\
\Lambda_{2,i+N}^-&=0,  \\ 
r_{2,i+N}&=w_i+\delta_i.
\end{cases}
\label{eq:traffic}
\end{align}
and 
\begin{align}
q_{1,i}=\frac{\Lambda_{1,i}^+}{r_{1,i}+\Lambda_{1,i}^-}; \quad
q_{2,i+N}=\frac{\Lambda_{2,i+N}^+}{r_{2,i+N}+\Lambda_{2,i+N}^-},
\end{align}
where $i \in \{1,\dots,N\}$. The $q_{1,i}$ denotes the probability that the $i^{th}$ queue has at least one job. Since queues $1,\dots,N$ are WSs which cannot store EPs, the $q_{2,i}=0$ for all $i \in \{1,\dots,N\}$. Moreover, the $q_{2,i+N}$ denotes the probability that the $(i+N)^{th}$ queue has at least one EP. Analogously, the $q_{1,i+N}=0$ for all $i \in \{1,\dots,N\}$ since queues $N+1,\dots,2N$ are ESs at which jobs cannot wait to be served. Note that, $\lambda_i$, $w_i$ and $\delta_i$ are the parameters of the $i^{th}$ ES which is the $(i+N)^{th}$ queue.  We use subscript $i+N$ in $\Lambda_{2,i+N}^+$, $\Lambda_{2,i+N}^-$ and $r_{2,i+N}$  since they are related to the traffic equation of the $(i+N)^{th}$ queue.

\begin{Theorem}
\label{theorem1}
Let $K(t)=(K_1(t),\dots,K_N(t))$ and $B(t)=(B_1(t),\dots,B_N(t))$. If a non-negative solution of traffic equations given in (\ref{eq:traffic}) exists such that $0< q_{c,i} <1$ for $c=1,2$ and $i=1,\dots,N$, then the PFS:
\begin{align}
\lim_{t \rightarrow \infty} \Pr [(K(t),B(t))=(k_1,\dots,k_N,b_1,\dots,b_N)] =
\prod_{i=1}^N q_{1,i}^{k_i}(1-q_{1,i}) q_{2,i+N}^{b_i}(1-q_{2,i+N}),
\end{align}
exists.
\end{Theorem} 
\begin{proof}
Proof is given in \cite{gelenbe1992stability}.
\end{proof}
In the steady-state, the PFS is the joint probability distribution of the number of jobs at WSs and the number of EPs at ESs. 

\section{The EPN and its Optimization Problem}
\label{opt_problem_with_GA}

In this paper, we consider an optimization problem. The objective is to minimize a cost function consisting of the average response time for jobs and energy loss, by choosing an optimal energy distribution. 
We propose the composite cost function that needs to be minimized:
\begin{equation}
C= W+ E,
\end{equation}
where 

\begin{eqnarray}
W&=& \sum_{i=1}^N \frac{\lambda_i}{\sum_{i=1}^N \lambda_i} \frac{(r_{1,i}+\Lambda_{1,i}^-)^{-1}}{1-q_{1,i}}\nonumber \\
&=&\frac{1}{\sum_{i=1}^N \lambda_i} \sum_{i=1}^N \frac{q_{1,i}}{1-q_{1,i}}, \\
E&=& \sum_{i=1}^N q_{2,i+N} \delta_i +q_{2,i+N} w_i (1-q_{1,i}). 
\end{eqnarray}

By Little's Law, the $W$ represents the overall average response time of jobs waiting to be served at each of the WSs.  The $E$ represents the total rate of EPs lost in the network. The first term of $E$ regards the rate of EPs loss due to storage leakages and the second term of $E$ regards the rate of EPs lost in the network when EPs are sent to idle WSs. 
With a specific constraint for this case with $\sum_{i=1}^N p_i =1$, for $1 \leq i \leq N$, we have:
\begin{eqnarray}
q_{1,i}&=&\frac{\lambda_i}{q_{2,i+N} \cdot w_i[\frac{1-\sum_{s=1}^\infty q_{1,i}^s \cdot  \pi_{i,s}}{1-q_{1,i}}]}, \label{eq:q1}\\
q_{2,i+N} &=& \frac{\gamma p_i}{w_i + \delta_i}. 
\end{eqnarray}

The optimization problem  is to choose optimal $p=(p_1,\dots,p_N)$ so as to minimize the composite cost function $C$ for the given value of network parameters, i.e. $\gamma$, $w_i$, $\lambda_i$, $\delta_i$ and $u_i$ at each  $i^{th}$ WS and $i^{th}$ ES.

In the following sections, we first make a geometric assumption regarding the probability distribution $\pi_{i,s}$ in order to obtain convenient analytical results. Under this assumption, we analyze the network in the case of the  $N=2$, and $N \geq 3$ respectively in Section \ref{ana_problem_with_GA}.  We also consider  Proposition \ref{prop1}  that guarantees the existence of a globally optimal solution under the geometric assumption.  

A special case with sufficiently large $\gamma$ EPs/sec is discussed in Section \ref{problem_special}. With a proper approximation, one simplified cost function $\hat{C}$ is given to solve the optimization problem.

\section{Analysis of the Optimization Problem}
\label{ana_problem_with_GA}

To simplify the analysis, we make a geometric assumption regarding the probability distribution $\pi_{i,s}$. We assume that 
\begin{eqnarray}
\pi_{i,s} = \frac{(1-u_i)u_i^s}{u_i},
\label{eq:GA}
\end{eqnarray} 
where $0<u_i<1$ is a real number. Note that  we have 
$
\sum_{s=1}^\infty  (1-u_i)u_i^{s-1} =1.
$
The $\pi_{i,s}$ is the probability distribution which regards the number of jobs that can be served by one EP. With the geometrical assumption, it is convenient to analyze the optimization problem analytically.

Substituting (\ref{eq:GA}) into (\ref{eq:q1}), we have 

\begin{equation}
\begin{aligned}
q_{1,i} =& \frac{\lambda_i}{q_{2,i+N} w_i} [\frac{1-\sum_{s=1}^N q_{1,i}^s u_i^s \frac{1-u_i}{u_i}}{1-q_{1,i}}]^{-1}  \\
=& \frac{\lambda_i}{q_{2,i+N} w_i} [\frac{1-\frac{(1-u_i)q_{1,i}}{1-u_i q_{1,i}}}{1-q_{1,i}}]^{-1}  \\
=&\frac{\lambda_i}{u_i \lambda_i + q_{2,i+N}w_i}.\\
\end{aligned}
\end{equation}
Then the cost functions $W$ and $E$ are, respectively:
\begin{eqnarray}
W&=& \sum_{i=1}^N \frac{1}{\lambda^+} \cdot \frac{\lambda_i}{\lambda_i(u_i-1)+\sigma_i \gamma p_i},\\
E
 &=& \gamma-\lambda^+ + \sum_{i=1}^N \frac{\lambda_i^2 u_i}{\lambda_i u_i + \sigma_i \gamma p_i},
\end{eqnarray}
where 
$
\sigma_i = \frac{w_i}{w_i + \delta_i},
$
denotes the energy efficiency regarding to leakages at $E_i$. Moreover, $\lambda^+ = \sum_{i=1}^N \lambda_i$ is a constant.
Note that it is an optimization problem subject to the constraint $\sum_{i=1}^N p_i =1$. Therefore, for any $p_j$, it follows  $p_j=1-\sum_{i=1,i\neq j}^N p_i$ and we have 
\begin{equation}
\frac{\partial p_j}{\partial p_i} = -1, ~\forall i \neq j. 
\end{equation}

\subsection{The case of $N=2$}

In the case of $N=2$, we have two pairs of WSs and ESs. The composite cost function can be written as 
\begin{align}
&C= \gamma - \lambda^+ + \frac{1}{\lambda^+} \frac{\lambda_1}{\lambda_1(u_1-1)+\sigma_1 \gamma p_1}+ \frac{\lambda_1^2 u_1}{\lambda_1 u_1 + \sigma_1 \gamma p_1} \nonumber \\
&+ \frac{1}{\lambda^+} \frac{\lambda_2}{\lambda_2(u_2-1)+\sigma_2 \gamma (1-p_1)}+ \frac{\lambda_2^2 u_2}{\lambda_2 u_2 + \sigma_2 \gamma (1-p_1)}, 
\end{align} 
where $p_2=1-p_1$. The constraint is removed by replacing $p_2$ with $(1-p_1)$. The derivative of $C$ with respect to $p_1$ is 
\begin{align}
&\frac{d C}{d p_1} = \frac{-\lambda_1 \sigma_1 \gamma }{\lambda^+[\lambda_1 (u_1-1)+\sigma_1 \gamma p_1]^2}+\frac{- \lambda_1^2 \sigma_1 u_1 \gamma }{[\lambda_1 u_1+\sigma_1 \gamma p_1]^2} \nonumber +\\
&\frac{\lambda_2 \sigma_2 \gamma }{\lambda^+[\lambda_2 (u_2-1)+\sigma_2 \gamma(1- p_1)]^2}+\frac{ \lambda_2^2 \sigma_2 u_2 \gamma }{[\lambda_2 u_2+\sigma_2 \gamma (1- p_1)]^2}.
\end{align} 
Let $dC / d p_1 =0 $, we have a function:

\begin{align}
\frac{\lambda_2 \sigma_2}{[\lambda_2 (u_2-1)+\sigma_2 \gamma - \sigma_2 \gamma p_1]^2}+\frac{\lambda_2^2 \sigma_2 u_2 \lambda^+}{[\lambda_2 u_2 + \sigma_2 \gamma-\sigma_2 \gamma p_1]^2}  =\nonumber \\
\frac{ \lambda_1 \sigma_1}{[\lambda_1 (u_1-1)+\sigma_1 \gamma p_1]^2}+\frac{\lambda^2_1 \sigma_1 u_1 \lambda^+ }{[\lambda_1 u_1 + \sigma_1 \gamma p_1]^2}. 
\label{eq:2pair}
\end{align}
By solving (\ref{eq:2pair}) with respect to $p_1 \in [0,1]$, the optimal solutions $p_1^*$ and $p_2^*=1-p_1^*$ are obtained. 
In order to illustrate the analytically obtained optimal solution, we consider a numerical example with two pairs of WSs and ESs, and their parameters are shown in Table \ref{table:1}. 

\begin{table}[h]
\renewcommand{\arraystretch}{1.3}
\caption{Parameters of the EPN with $N=2$}
\centering
\begin{tabular}{cc||cc}
\hline
Parameters & Values & Parameters & Values \\
 \hline
$\gamma$ & $150$ EPs/sec & $\lambda_1,~\lambda_2, $ & $50,60$ jobs/sec  \\
$\delta_1,~\delta_2$& $10, 6$ EPs/sec & $w_1,~w_2,$&$100, 80$ EPs/sec
 \\ $u_1,~u_2$& $0.2, 0.2$ &  \\
\hline
\end{tabular}
\label{table:1}
\end{table}  
To guarantee the $q_{c,i}<1$ for all $c,i$, the condition 
\begin{equation}
\frac{\lambda_i (1-u_i)}{\gamma \sigma_i} < p_i <\frac{w_i +\delta_i}{\gamma}
\end{equation}  
must hold. Numerical conditions:
\begin{eqnarray*}
0.2933 < p_1 < 0.7333 \mbox{ and } 0.3400 < p_2 <0.5733,
\end{eqnarray*}
and 
$
p_1 + p_2=1
$
must hold. Then we calculate the values of the cost function $C$ with all $(p_1,p_2)$ and compare them to the optimal solution. The results are shown in Figure \ref{fig:P1_1} in which the x-axis is $p_1$ while $p_2$ follows from $p_2=1-p_1$.  

\begin{figure}[h]
\centering
  \includegraphics[width=0.7\linewidth]{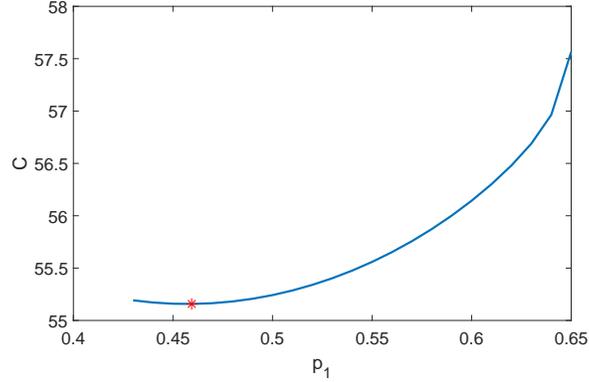} 
  \captionof{figure}{The cost function of the system with all $(p_1,p_2)$ pairs. The red point designates the optimal solution by solving (\ref{eq:2pair}). Note that the range $p_i$ for all $i$ is not $[0,1]$ due to constraints.}
  \label{fig:P1_1}
\end{figure}
The theoretical result is given by solving (\ref{eq:2pair}). The optimal solution $(p_1^*,p_2^*)=(0.4594,0.5406)$ produces the minimal cost function $C= 55.1577$, $E=55.1177$ EPs/sec and $W=0.0400$ seconds. 
\begin{Comment}
Moreover we observed that the maximum values of the W and E are $0.6609$ seconds and $56.9040$ EP/s. Average response time of jobs waiting to be served at each of the WSs is reduced by 0.6209 seconds. We can also save $1.7863$ EPs in every second. Additionally, in the numerical simulation, we found the minimal values of $W$ and $E$ when we try to minimize only $W$ or $E$ instead of minimizing the composite cost function $C$. They are $0.0399$ seconds and $55.1177$ EPs/sec respectively. The optimal solutions obtained by minimizing the composite cost function are close to the real minimal values.
\end{Comment}

\subsection{The case of $N \geq 3$}
In the case $N \geq 3$, the composite function $C$ can be rearranged to 
\begin{align}
& C=\gamma - \lambda^+ + \frac{1}{\lambda^+} \frac{q_{1,1}}{1-q_{1,1}}+ \lambda_1 u_1 q_{1,1} \nonumber + \sum_{i=2}^N \frac{1}{\lambda^+} \frac{q_{1,i}}{1-q_{1,i}}+ \lambda_i u_i q_{1,i}.
\end{align}
The partial  derivative of $C$ with respect to $q_{1,i}$ for $i \geq 2$ is 
\begin{align}
\frac{\partial C}{\partial q_{1,i}} = \frac{1}{\lambda^+(1-q_{1,i})^2}+\lambda_i u_i  \nonumber +(\frac{1}{\lambda^+(1-q_{1,1})^2}+\lambda_1 u_1)\big(-\frac{\sigma_1 \lambda_i}{\sigma_i \lambda_1}( \frac{q_{1,1}}{q_{1,i}})^2 \big). 
\end{align}
Let $\partial C / \partial q_{1,i} = 0$, we have $N-1$ equations for $i=2,\dots,N$:
\begin{align}
\frac{1}{\lambda^+(1-q_{1,i})^2}+\lambda_i u_i = f_i (\frac{1}{q_{1,i}})^2 . \label{eq:3pair}
\end{align}
where
\begin{equation}
f_i = (\frac{1}{\lambda^+(1-q_{1,1})^2}+\lambda_1 u_1)\frac{\sigma_1 \lambda_i}{\sigma_i \lambda_1}(q_{1,1})^2.
\end{equation}
Rearrange (\ref{eq:3pair}), we have
\begin{align}
 q_{1,i}^4-2 q_{1,i}^3+&\frac{\lambda_i u_i \lambda^+ +1 - f_i}{\lambda_i u_i \lambda^+} q_{1,i}^2
 +\frac{2f_i}{\lambda_i u_i} q_{1,i}-\frac{f_i}{\lambda_i u_i} = 0. \label{eq:3pair_quartic}
\end{align}
The optimal solution is obtained by solving $N-1$ equations for $i \in \{2,\dots,N \}$ given in (\ref{eq:3pair_quartic}) and one constraint:
\begin{eqnarray}
\sum_{i=1}^N p_i = \sum_{i=1}^N \frac{\lambda_i}{\sigma_i \gamma}(\frac{1}{q_{1,i}}-u_i)=1,
\end{eqnarray}
simultaneously. Now we present a sufficient condition which guarantees existence of a globally optimal solution in Proposition \ref{prop1}.

\begin{Proposition}
\label{prop1}
Under the geometric assumption, a global optimal solution of the composite cost function $C$ exists if the PFS given in Theorem \ref{theorem1}
exists.
\end{Proposition}

\begin{proof}
In \cite{gelenbe1992stability}, Gelenbe and Schassberger prove the PFS exists if $q_{c,i}<1$ for all $c,~i$ holds. Thus conditions 
\begin{eqnarray}
q_{1,i}=\frac{\lambda_i}{\lambda_i u_i + \gamma \sigma_i p_i} < 1, \quad
q_{2,i+N} = \frac{\gamma p_i}{w_i+\delta_i} <1,
\end{eqnarray} 
must hold. 
Hence, the PFS of the EPN exists if the condition
\begin{eqnarray}
\frac{\lambda_i (1-u_i)}{\gamma \sigma_i} < p_i <\frac{w_i +\delta_i}{\gamma},
\end{eqnarray} 
holds for all $i$.
Then the Hessian matrix $\nabla_{pp} C$ is an diagonal matrix with $N$ diagonal entries 
\begin{eqnarray}
\frac{\partial C^2}{\partial^2 p_i} = \frac{2\lambda_i \gamma^2 \sigma_i^2}{\lambda^+ [\gamma \sigma_i p_i - \lambda_i(1-u_i)]^3}+\frac{2\lambda_i^2 u_i \sigma_i^2 \gamma^2}{\lambda_i u_i + \sigma_i \gamma p_i} >0, \mbox{ for } i=1,\dots,N ,
\end{eqnarray}
if the PFS exists. The cost function $C$ is a strictly convex function with respect to $p$ since the Hessian $\nabla_{pp} C \succ 0$. 

The following proof is by contradiction.
Assume $p^a$ is a minimal point of the cost function and $p^b$ is another minimal point of the cost function such that $C(p^a)=C(p^b)$. Since the cost function is strictly convex, we have 
\begin{equation}
C(\alpha p^a +(1-\alpha) p^b) < \alpha C(p^a) +(1-\alpha) C(p^b)=C(p^a),
\end{equation} 
where $0 \leq \alpha \leq 1$ is a real number. This contradicts to the assumption $p^a$ is the minimal point of the cost function. Therefore, there exists a globally optimal solution of the composite cost function $C$ if the PFS exits. 
\end{proof}

In order to illustrate the analytically result with $N \geq 3$, we consider a numerical example without loss of generality.  There are three pairs of WS and ES nodes, and their parameters shown in Table \ref{table:2}.

\begin{table}[h]
\renewcommand{\arraystretch}{1.3}
\caption{Parameters of the EPN with $N=3$}
\label{table:2}
\centering
\begin{tabular}{cc||cc}
\hline
Parameters & Values & Parameters & Values\\
 \hline
$\gamma$ & $150$ EPs/sec & $\lambda_1,~\lambda_2, ~\lambda_3$ & $50, 30, 10$ jobs/sec  \\
$\delta_1,~\delta_2,~\delta_3$& $10, 8,6$ EPs/sec &$w_1,~w_2,~w_3$&$100, 80,50$ EPs/sec\\
$u_1,~u_2,~u_3$& $0.2, 0.2,0.2$ \\
\hline
\end{tabular}
\end{table}   
The constraint and numerical conditions:   
\begin{eqnarray*}
&0.2933 < p_1 < 0.7333, ~ &0.1760 < p_2 <0.5867, \\ &0.0597<p_3<0.3733, ~&p_1 + p_2 +p_3 =1,
\end{eqnarray*}  
must hold to guarantees $q_{c,i} <1$ for all $c,i$. Then we calculate  the values of the cost function $C$ with all $(p_1,p_2,p_3)$ and compare them to the optimal solution. The results are shown in Figure \ref{fig:P2_1} in which the x-axis and y-axis are $p_1$ and $p_2$, while $p_3$ follows from $p_3=1-p_1-p_2$.    

\begin{figure}[h]
\centering
  \includegraphics[width=0.7\linewidth]{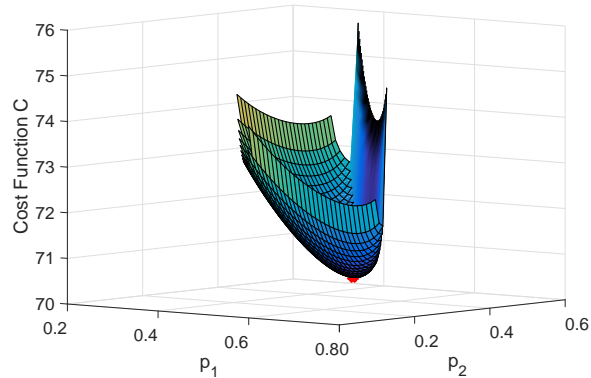} \captionof{figure}{The cost function $C$ of the system with all $(p_1,p_2)$ pairs. The red point designates the optimal solution. Note that the range $p_i$ for all $i$ is not $[0,1]$ due to constraints.}
  \label{fig:P2_1}
\end{figure}

The theoretical optimal solution $(q_{1,1}^*,q_{1,2}^*,q_{1,3}^*)=(0.5847,0.5835,0.5827)$ gives energy distribution probability  $(p_1^*,p_2^*,p_3^*)=(0.5538,0.3330,0.1132)$. Note that all $p_i$ satisfies the numerical conditions and the constraint. It yields $C=70.5602$, $E=70.5135$ EPs/sec and $W=0.0467$ seconds. 
\begin{Comment}
We observed that the maximum values of the W and E are $6.1308$ seconds and $73.8681$ EP/s. The average response time of jobs waiting to be served at each of the WSs is reduced by $6.0841$ seconds. We can also save $3.3546$ EPs in every second. In the numerical simulation, we also found the minimal values of $W$ and $E$ when we try to minimize $W$ or $E$ only. They are $0.0429$ seconds and $70.5133$ EPs/sec respectively. The difference between our optimal solutions and real minimal values are $0.0038$ seconds and $0.0002$ EPs/sec.
\end{Comment}

\section{A Special Case of the Optimization Problem}
\label{problem_special}
In this section, we discuss a special case when $q_{1,i}$ is small (close to 0) for all $i$. It implies that the rate of total harvesting energy is sufficiently large. 

When $q_{1,i}$ is small, the cost function $W$ regarding the overall delay is approximated to 
\begin{equation}
W=\frac{1}{\lambda^+} \sum_{i=1}^N \frac{q_{1,i}}{1-q_{1,i}} \approx  \frac{1}{\lambda^+} \sum_{i=1}^N q_{1,i}^2+q_{1,i}.
\end{equation}
The composite cost function becomes
\begin{eqnarray}
C =  \gamma-\lambda^+ + \sum_{i=1}^N \frac{1}{\lambda^+} q_{1,i}^2+(\frac{1}{\lambda^+} + u_i \lambda_i) q_{1,i}.
\end{eqnarray}
The  cost function $C$ is a monotonic increasing function in the range of $0 < q_{1,i} < 1$ for all $i$. Therefore, we proposed a simplified composite cost function based on Lagrangian multiplier:
\begin{align}
\hat{C} 
&=\sum_{i=1}^N \frac{\lambda_i}{\lambda_i u_i + \sigma_i \gamma p_i} + \beta (\sum_{i=1}^N p_i -1),
\end{align}
where $\beta$ is a real number. Let $ \partial \hat{C} / \partial p_i =0 $, we derive solution:
\begin{eqnarray}
\hat{p}_i = \frac{\sqrt{\frac{\lambda_i}{\sigma_i}}}{\sum_ {i=1}^N \sqrt{\frac{\lambda_i}{\sigma_i}}} \big( 1+ \sum_{i=1}^N \frac{\lambda_i u_i}{\sigma_i \gamma} \big) - \frac{\lambda_i u_i}{\sigma_i \gamma}. 
\label{eq:estimated_p}
\end{eqnarray}
which fulfils the constraint $\sum_{i=1}^N p_i =1$. 

We consider a numerical example with sufficiently large energy harvesting rate, $\gamma$ to illustrate the analytically obtained optimal solution of the optimization problem when $q_{1,i}$ is small. There are three pairs of WS and ES nodes and their parameters are shown in Table \ref{table:3}. 

\begin{table}[h]
\renewcommand{\arraystretch}{1.3}
\caption{Parameters of the EPN with large $\gamma$}
\label{table:3}
\centering
\begin{tabular}{cc||cc}
\hline
Parameters & Values & Parameters & Values\\
 \hline
$\gamma$ & $230$ EPs/sec & $\lambda_1,~\lambda_2, ~\lambda_3$ & $5, 6, 5$ jobs/sec  \\
$\delta_1,~\delta_2,~\delta_3$& $10, 10,25$ EPs/sec &$w_1,~w_2,~w_3$&$100, 80,65$ EPs/sec\\
$u_1,~u_2,~u_3$& $0.2, 0.2,0.2$ \\
\hline
\end{tabular}
\end{table} 

The numerical conditions
\begin{eqnarray*}
&0.0191 < p_1 < 0.4783, ~ &0.0235 < p_2 <0.3913, \\ &0.0241<p_3<0.3913, ~&p_1 + p_2 +p_3 =1,
\end{eqnarray*}  
must hold to guarantee $q_{c,i}<0$ for all $c,i$. We calculate the value of the cost function $C$ with all $(p_1,p_2,p_3)$ and compare then to the solution $\hat{p}$ given by (\ref{eq:estimated_p}). Then we also compare the $\hat{C}$ given by $\hat{p}$ and minimal $C^*$ observed in Figure \ref{fig:P3_1}. The results are shown in Figure \ref{fig:P3_1} in which the x-axis and y-axis are $p_1$ and $p_2$.

\begin{figure}[h]
\centering
  \includegraphics[width=0.7\linewidth]{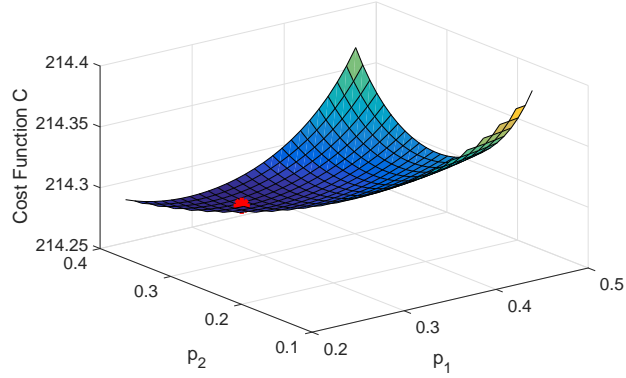} 
  \captionof{figure}{The cost function of the system with all $(p_1,p_2)$ pairs. The red point designates the  solution by solving (\ref{eq:estimated_p}). Note that the range $p_i$ for all $i$ is not $[0,1]$ due to constraints.}
  \label{fig:P3_1}
\end{figure}

The theoretical solution given by (\ref{eq:estimated_p}) is $ (\hat{p_1},\hat{p_2},\hat{p_3})=(0.3100,0.3429,0.3471)$. It produces the cost function $\hat{C}=214.2789$. 
At this point, we have $q_{1,1}=0.0760$, $q_{1,2}=0.0845$ and $q_{1,3}=0.0852$. They are all close to zero. Moreover, the minimal cost function observed in Figure \ref{fig:P3_1} is $C^*=214.2784$  and the error between $C^*$ and $\hat{C}$ is $0.005$.
\begin{Comment}
The theoretical solution gives optimal $\hat{W}=0.00167$ seconds and $\hat{E}=214.2622$ EPs/sec. In feasible region, the observed maximum values are $0.00245$ seconds and $214.3678$ EPs/sec respectively. The optimal solution reduces the average response time of jobs by $0.00078$ seconds and saves $0.1056$ EPs in every second. Moreover, the minimal $W^*$ and $E^*$ observed in Figure \ref{fig:P3_1} are $0.0167$ seconds and $214.2617$ EPs/sec. The difference between $\hat{W}$ and $W^*$, and $\hat{E}$ and $E^*$ are small. 
\end{Comment}
\section{Conclusions}
\label{conclusions}
In this paper, we have considered an EPN model represented in  G-networks with batch removal. In the network, the jobs can be executed and removed under the effect of EPs when EPs arrive at a WS. We assume the size of the batch that one EP can remove follows a known probability distribution. Moreover, the ES may lose energy through leakage, and idle WSs will also consume energy.

We consider an optimization problem of the case that neither jobs nor EPs will be moved between WSs or ESs respectively. Since the total harvested energy rate is fixed, we consider an optimal energy distribution among ESs. The optimal distribution minimizes a cost function, which is a combination of the average response time of jobs and the rate of energy lost. This problem is solved analytically for a special class of probability distributions for the number of jobs processed with one EP. There is a proposition to guarantee the existence of the globally optimal solution. Moreover, we also discuss a special case when the fixed energy harvesting rate is sufficiently large. The original cost function is simplified to give an approximated solution.

Future work will investigate the minimization of average response time and energy loss with new control variables in more general EPN models. The probability transition matrix regarding movements of jobs in WSs and EPs in ESs can be considered. More cases regarding energy harvesting rate will be considered. Moreover, the the EPN paradigm with optimal energy distribution will be considered in some real-world applications.  
%
%
%
\bibliographystyle{splncs04}
\bibliography{my_references}

\end{document}